\documentclass{article}
\usepackage{graphicx} 
\usepackage{amsmath,appendix}
\usepackage{amssymb,stmaryrd}
\usepackage{amsfonts,url,comment}

\usepackage{amsthm}
\usepackage{tabularray}
\usepackage[linesnumbered,ruled,vlined]{algorithm2e}
\usepackage{graphicx,color}
\usepackage{subfigure}
\usepackage[linesnumbered,ruled,vlined]{algorithm2e}

\newcommand{\tconv}{\text{tconv\,}}

\newtheorem{thm}{Theorem}

\newtheorem{definition}[thm]{Definition}

\newtheorem{remark}[thm]{Remark}

\newcommand{\RR}{\mathbb{R}}
\newcommand{\R}{\mathbb{R}}

\title{Imputing phylogenetic trees using tropical polytopes over the space of phylogenetic trees}
\author{Ruriko Yoshida}
\date{}

\begin{document}

\maketitle

\begin{abstract}
    When we apply comparative phylogenetic analyses to genome data, it is a well-known problem and challenge that some of given species (or taxa) often have missing genes.  In such a case, we have to impute a missing part of a gene tree from a sample of gene trees.   In this short paper we propose a novel method to infer a missing part of a phylogenetic tree using an analogue of a classical linear regression in the setting of tropical geometry.  In our approach, we consider a tropical polytope, a convex hull with respect to the tropical metric closest to the data points.  We show a condition that we can guarantee that an estimated tree from our method has at most four Robinson–Foulds (RF) distance from the ground truth and computational experiments with simulated data show our method works well.  
\end{abstract}
\section{Introduction}

Due to a new technology, today we are able to generate sequences from genome with lower cost.  However, at the same time, we have a great challenge to analyze large scale datasets from genome sequences.  In phylogenomics, a new field which applies tools from phylogenetics to genome datasets, we often conduct comparative phylogenetic analyses, that is, to compare evolutionary histories among a set of taxa between different genes from genome (for example, see \cite{Koonin}). However, we often face the problem in this process that some taxa in the dataset have missing gene(s) \cite{8554124}.  When it happens, systematists infer missing part of a gene tree from other gene trees using supervised learning method, such as linear regression model.  

A phylogenetic tree is a weighted tree which represents evolutionary history of a given set of taxa (or species).  In a phylogenetics, leaves represent species or taxa in the present time which we can observed, and internal nodes in the tree, which represent ancestors of the species, do not have any labels.  A gene tree is a phylogenetic tree reconstructed from an alignment of a gene in a genome.  Gene trees with the same set of species or taxa do not have to have the same tree topology since each gene might have different mutation rates due to the selection pressures, etc \cite{coalescent}.  In a comparative phylogenetic analysis, we often compare gene trees (for example, we compare how they are different, how their mutation rates are different, and often we are interested in inferring the species tree).  

To infer a missing part of a gene tree, we often apply a supervised method to regress the missing part.  In this process, first, we compute an unique vector representation of each gene tree.  Then we infer the missing components of the vector of the tree from the vectors computed from other gene trees in a dataset using a regression model, such as a multiple linear regression \cite{8554124}.

However, a set of all such vectors realizing all possible phylogenetic trees, which is called a {\em space of phylogenetic trees}, is not Euclidean.  In fact, a space of phylogenetic trees is an union of polyhedral cones with a large co-dimension, so this is not even convex in terms of Euclidean metrics.  Therefore, it is not appropriate to apply classical regression models, such as linear regression or Neural Networks, since they assume convexity in terms of Euclidean geometry.  Thus, in this short paper, we propose an analogue of a classical multiple linear regression in the setting of tropical geometry with the max-plus algebra: an application of tropical polytopes to infer the missing part of a phylogenetic tree.  

{\em Equidistant trees} are used to model gene trees under the multi-species coalescent model \cite{coalescent}. Therefore, in this paper, we focus on an equidistant tree, which is a rooted phylogenetic tree such that the total weight on an unique path from its root to each leaf is the same, and we focus on the space of all possible equidistant trees.  It is well-known that the space of all possible equidistant trees is a {\em tropical Grassmannian}, which means that it is a tropically linear space with respect to the {\em tropical metric} \cite{AK,SS,10.1093/bioinformatics/btaa564}.  Therefore, with the tropical metric with the max-plus algebra, we can conduct statistical analyses using tropical linear algebra, analogue of a classical linear algebra.  In fact, there has been much development in statistical learning over the space of phylogenetic trees using tools from tropical geometry \cite{10.1093/bioinformatics/btaa564,YZZ,anthea,ArdilaKlivans,LSTY,Yoshida,Yoshida2}.   

Since a tropical polytope is tropically convex and since the space of equidistant trees is tropically convex, if all vertices are equidistant trees, then a tropical polytope is contained in  the space of equidistant trees.  Thus, in this paper, we propose to use a tropical polytope over the space of equidistant trees to infer missing part of a phylogenetic trees. Our proposed method has basically four main step: (1) compute induced trees on the set of leaves which we observe in $T$, a tree with missing leaf (leaves) from a training set; (2) compare $T$ with these induced trees; (3) compute a tropical polytope with trees with full set of leaves whose induced trees have closest tree topologies with $T$; and (4) project $T$ onto the tropical polytope computed in Step (3).

In Section \ref{sec:tropical:basic} we discuss basics from tropical geometry and in Section \ref{sec:phylo:basic}, we discuss basics from phylogenetics. In Section \ref{sec:method}, we show our novel method to impute a missing part of a phylogenetic tree.  Then, in Section \ref{sec:theory}, we show a theoretical condition of $T$ that the worst case scenario for the estimated tree via our method has the {\em Robinson-Foulds distance} at most 4.  Then Section \ref{sec:comp} shows computational experiments of our method against other methods including a multiple linear regression and our method performs well.  

\section{Basics in Tropical Geometry}\label{sec:tropical:basic}

In this section, we discuss basics from tropical geometry.  We consider the {\em tropical projective torus},  $\mathbb R^e \!/\mathbb R {\bf 1}$ where ${\bf 1}:=(1, 1,\ldots , 1)$ is the vector with all ones in $\mathbb{R}^e$.  Basically this means that any vectors in $\mathbb R^e \!/\mathbb R {\bf 1}$ is invariant with ${\bf 1}$, i.e., $(v_1+c, \ldots , v_e+c) = (v_1, \ldots , v_e) = v$ for any element $v:=(v_1, \ldots , v_e) \in \mathbb R^e \!/\mathbb R {\bf 1}$.  
For more details, see \cite{ETC} and \cite{MS}.

Under the tropical semiring $(\,\mathbb{R} \cup \{-\infty\},\oplus,\odot)\,$, the tropical arithmetic operations of addition and multiplication are defined as:
$$a \oplus b := \max\{a, b\}, ~~~~ a \odot b:= a + b ~~~~\mbox{  where } a, b \in \mathbb{R}\cup\{-\infty\}.$$
For any scalars $a,b \in \mathbb{R}\cup \{-\infty\}$ and for any vectors $x = (x_1, \ldots ,x_e),\; y= (y_1, \ldots , y_e) \in \mathbb R^e \!/\mathbb R {\bf 1}$, we have tropical scalar multiplication and tropical vector addition defined as:
$$a \odot x \oplus b \odot y := (\max\{a+x_1,b+y_1\}, \ldots, \max\{a+x_e,b+y_e\}).$$

\begin{definition}\label{def:polytope}
Suppose we have a set $S \subset \mathbb R^e \!/\mathbb R {\bf 1}$. If 
\[
a \odot x \oplus b \odot y \in S
\]
for any $a, b \in \R$ and for any $x, y \in S$, then $S$ is called {\em tropically convex}.
Suppose we have a finite subset $V = \{v^1, \ldots , v^s\}\subset \mathbb R^e \!/\mathbb R {\bf   1}$.  Then, the smallest tropically-convex subset containing $V$ is called the {\em tropical convex hull} or {\em tropical polytope} of $V$. $\mathrm{tconv}(V)$ can also be written as:
$$ \mathrm{tconv}(V) = \{a_1 \odot v^1 \oplus a_2 \odot v^2 \oplus \cdots \oplus a_s \odot v^s \mid  a_1,\ldots,a_s \in \R \}.$$
A {\em tropical line segment}, $\Gamma_{v^1, v^2}$, between two points $v^1, \, v^2$ is a tropical convex hull of $\{v^1, v^2\}$. 
\end{definition}

\begin{remark}
    By the definition, if a set $S \subset \mathbb R^e \!/\mathbb R {\bf 1}$ is tropically convex, then a tropical line segment between any two points in $S$ must be contained in $S$.
\end{remark}

\begin{definition}
\label{eq:tropmetric} 
For any points $v:=(v_1, \ldots , v_e), \, w := (w_1, \ldots , w_e) \in \mathbb R^e \!/\mathbb R {\bf 1}$,  the {\em tropical metric}), $d_{\rm tr}$, between $v$ and $w$ is defined as:
\begin{equation*}
d_{\rm tr}(v,w)  := \max_{i \in \{1, \ldots , e\}} \bigl\{ v_i - w_i \bigr\} - \min_{i \in \{1, \ldots , e\}} \bigl\{ v_i - w_i \bigr\}.
\end{equation*}
\end{definition}

\begin{definition}\label{def:proj}
Let $V:= \{v^1, \ldots, v^s\} \subset \mathbb{R}^e/{\mathbb R} {\bf 1}$ and let $P = \tconv(v^1, \ldots, v^s)\subseteq \mathbb R^{e}/\RR{\bf 1}$ be a tropical polytope with its vertex set $V$.
For $x \in \mathbb R^{e}/\RR{\bf 1}$, let
\begin{equation}\label{eq:tropproj} 
\pi_{P} (x) \!:=\! \bigoplus\limits_{l=1}^s \lambda_l \odot  v^l, ~ ~ {\rm where} ~ ~ \lambda_l \!=\! {\rm min}\{x-v^l\}.
\end{equation}
Then $\pi_{P} (x)$ is a projection onto $P$ with the property such that
\[
d_{\rm tr}(x, \pi_{P} (x))  \leq d_{\rm tr}(x, y)
\]
for all $y \in P$.  
\end{definition}

\section{Basics in Phylogenetic Trees}\label{sec:phylo:basic}

Let $[m]:= \{1, \ldots , m\}$.  A phylogenetic tree $T$ on $[m]$ is a weighted tree of $m$ leaves with the labels $[m]$ and internal nodes in the tree do not have labels.  A subtree $T'$ in $T$ on $a \subset [m]$ is a subtree of $T$ with leaves $a$.  An {\em equidistant tree} on $[m]$ is a rooted phylogenetic tree on $[m]$ such that the total weight on the path from its root to each leaf $i$ in $[m]$ has the same distance for each $i \in [m]$.  In this paper, we assume on equidistant trees.

In order to conduct a statistical analysis, we have to convert a phylogenetic tree into a vector.  Now we discuss one way to convert a phylogenetic tree into a vector.  
\begin{definition}
    Suppose we have a {\em dissimilarity map} $D: [m] \times [m] \to \mathbb{R}$ such that 
    \[
    D(i, j) = \begin{cases}
        D(i, j) \geq 0 &\mbox{if } i \not = j\\
        0 &\mbox{otherwise.}
    \end{cases}
    \]
    If there exists a phylogenetic tree on $[m]$ such that $D(i, j)$ is the total weight on the unique path from a leaf $i \in [m]$ to a leaf $j \in [m]$, then we call $D$ as a {\em tree metric}.  
\end{definition}

\begin{remark}
    Since a tree metric of a phylogenetic tree on $[m]$ is symmetric and its diagonal is 0, we consider an upper triangular matrix of the tree metric and we consider the upper triangular matrix of the tree metric as a vector in $e = \binom{m}{2}$.  
\end{remark}

\begin{definition}
    Let $D: [m] \times [m] \to \RR$ be a metric over $[m]$, namely, $D$ is a map from $[m]\times [m]$ to $\RR$ such that
\begin{eqnarray}\nonumber
D(i, j) = D(j, i) & \mbox{for all } i, j \in [m]\\\nonumber
D(i, j) = 0 & \mbox{if and only if } i = j\\\nonumber
D(i, j) \leq D(i, k) + D(j, k) & \mbox{for all }i, j, k \in [m].
\end{eqnarray}

Suppose $D$ is a metric on $[m]$.  Then if $D$ satisfies 
\begin{eqnarray}
\max\{D(i, j), D(i, k), D(j, k)\} 
\end{eqnarray}
is attained at least twice for any $i,j,k \in [m]$, then $D$ is called an {\em ultrametric}.  
\end{definition}

It is well-known that if we have an ultrametric on $[m]$, then there is an unique equidistant tree on $[m]$ by the following theorem:
\begin{thm}[\cite{Buneman}]\label{thm:3pt}
Suppose we have an equidistant tree $T$ with a leaf label set $[m]$ and suppose $D(i, j)$ for all $i, j \in [m]$ is the distance from a leaf $i$ to a leaf $j$.  Then, $D$ is an ultrametric if and only if $T$ is an equidistant tree on $[m]$. 
\end{thm}

Therefore by Theorem \ref{thm:3pt}, in this paper, we consider the set of ultrametrics, $\mathcal{U}_m \subset \mathbb{R}^e /\mathbb{R} {\bf 1}$, on $m$ as the space of equidistant trees on $[m]$.

\begin{definition}\label{def:clade}
Let $a, b \subset [m]$ such that $a \cup b = [m]$ and $a \cap b = \emptyset$. Suppose we have an equidistant phylogenetic tree $T$ with the leave set $[m]$.  A {\em clade} of $T$ with leaves $a \subset [m]$ is an equidistant tree on $a$ constructed from $T$ by adding all common ancestral interior nodes of any combinations of only leaves $a$ and excluding common ancestors including any leaf from $[m] - a$   in $T$, and all edges in $T$ connecting to these ancestral interior nodes and leaves  $a$.  
\end{definition}

\begin{definition}
For a rooted phylogenetic tree, a {\em nearest neighbor interchange (NNI)} is an operation of a  phylogenetic tree to change its tree topology by picking three mutually exclusive leaf sets $X_1, X_2, X_3 \subset X$ and changing a tree topology of the clade, possibly the whole tree, consisting with three distinct clades with leaf sets $X_1$, $X_2$, and $X_3$.  
\end{definition}

\begin{remark}
    Since there are three possible ways of connecting three distinct clades,  NNI move possibly creates two new tree topologies on $[m]$.
\end{remark}

\begin{definition}
    Suppose we have rooted phylogenetic trees $T_1, T_2$ on $[m]$.  The {\em Robinson-Foulds (RF) distance} is the number of operations that the subtree of $T_1$ has the same tree topology as the subtree of $T_2$ by removing a leaf of $T_1$ and the subtree of $T_2$ has the same tree topology as the subtree of $T_1$ by removing a leaf of $T_2$. 
\end{definition}

\begin{remark}
    The RF distance is always divisible by 2.
\end{remark}

\begin{remark}\label{ref1}
    One can see clearly that the RF distance between two trees which is one NNI move a part is 2 since they differ only one internal edge in each tree.
\end{remark}

\section{Method}\label{sec:method}

In this section we introduce our method to infer a missing part of an equidistant tree using tools from tropical geometry.  
Let $RF(T_1, T_2)$ be the RF distance between $T_1$ and $T_2$.  The algorithm on our method is shown in Algorithm \ref{alg1}.

\begin{algorithm}[!h]
\caption{Imputation with a tropical polytope}\label{alg1}
\KwData{Equidistant tree $T'$ on $a \subset [m]$, and a sample of equidistant trees on $[m]$, $\{T_1, \ldots, T_n\}$}
\KwResult{Estimated imputed tree $\hat{T}$ on $[m]$}
Let $T'_i$ be a tree dropped leaves $[m] - a$ from each $T_1$ for $i\in \{1, \ldots , n\}$\;
Drop leaves $[m] - a$ from each $T_1$ for $i\in \{1, \ldots , n\}$\;
Set $S = \emptyset$\;
Set $d \gets m^2$\;
\For{$i = 1, \ldots , n-1$}{
\If{$d > RF(T'_i, T'_j)$}
{
    $d\gets RF(T'_i, T')$\;
}
}
\For{$i = 1, \ldots , n-1$}{
\If{$RF(T'_i, T') == d$}
{
    $S \gets T'_i \cup S$\;
}
}
Compute the tropical polytope $P$ of ultrametrics computed from trees $T_i$ for all trees $T'_i \in S$\;
Let $T_{tmp}$ be a tree attached leaves $[m] - a$ to the root of $T'$\;
Convert an ultrametric $u \in \mathcal{U}_m$ computed from $T_{tmp}$\;
Let $v$ be a projection of $u$ onto $P$\;
Realize an equidistant tree $\hat{T}$ from $v$ and return $\hat{T}$\;
\end{algorithm}

\section{Theoretical Results}\label{sec:theory}

Let $a, b \subset [m]$ such that $a \cup b = [m]$ and $a \cap b = \emptyset$.  
Let $\{T_1, \ldots , T_n\}$ be a sample of equidistant trees with $m$ leaves.  Let $T'_i$ for $i = 1, \ldots n$ be an equidistant tree with $a$ by dropping tips $b$ from $T_i$, i.e., $T'_i$ is an induced tree on $a$.  

\begin{thm}\label{th:main}
    Suppose $\{T_1, \ldots , T_n\}$ is a sample of equidistant trees with $[m]$ and let $T_i = T'_i \cup T''_i$ such that $T'_i$ is a subtree on $a$ which is an equidistant tree with $a$ by dropping tips $b$ from $T_i$ and $T''_i$ is an subgraph graph with $b$ by adding all common ancestral interior nodes of any combinations of only leaves $a$ and excluding common ancestors including any leaf from $[m] - a$ in $T_i$ for $i = 1, \ldots , n$. Suppose $T'_i$ and $T'$ have the same tree topology for $i = 1, \ldots , n$.    If $T''_i$ are clade in $T_i$ for $i = 1, \ldots , n$ and $T''$ is also a clade in $T$, then an estimated tree $\hat{T}$ via our method with the tropical polytope $P := \tconv(T_1, \ldots , T_n)$ and $T$ differ at most the RF distance = 4.
\end{thm}
\begin{proof}
    Since $T''_i$ are connected trees for $i = 1, \ldots , n$, $T''_i$ forms a clade in $T_i$ for $i = 1, \ldots , n$.  Also $T''$ is a connected tree, so that $T''$ is also a clade in $T$. This means that $T_i$ and $T_j$ for any $i, j \in \{1, \ldots , n\}$ have only one NNI move distance.  
    Since $T'_i$ and $T'$ have the same tree topology and since $T''$ is also a clade in $T$, $T_i$ and $T$ have only one NNI move distance.
    Note that $T_i$ and $T_j$ have at most the RF distance = 2 since $T_i$ and $T_j$ have just one NNI move difference, and so as with $T$.  Let $U_i$ is an ultrametric form a tree $T_i$ for $i = 1, \ldots , n$.  Then,
    take any tropical line segment $\Gamma_{u_i, u_j}$.  Since $T_i$ and $T_j$ have just one NNI move difference, by Theorem 8 in \cite{YC}, any tree topology of the tree realized from an ultrametric in $\Gamma_{u_i, u_j}$ has the same tree topology of $T_i$ or $T_j$.  Since $P$ is tropically convex, any point in $P$ is a tropically convex combination of $T_i$ for $i = 1, \ldots , n$.  Thus the tree topology of the tree realized by an ultrametric in $P$ has at most one NNI move different.  Since the estimate $\hat{T} \in P$, and any tree realized from an ultrametric in $P$ to $T$ has at most 2 NNI move difference.  Thus, we have the result.  
\end{proof}

\section{Computational Experiments}\label{sec:comp}
In this section, we apply our method to simulated data sets and compare its performance with the baseline model, which uses means of each missing element in an ultrametric computed from a tree, and multiple linear regression model.

\subsection{Simulated Data}
To assess a performance of our method, we use simulated datasets generated from the multi-species coalescent model using the software {\tt Mesquite} \cite{mesquite}.

Under the multi-species coalescent model, there are two parameters: species depth $SD$ and effective population size $N_e$.  In this paper we fix the effective population size as $N_e = 10,000$ and we vary $SD$ as we vary the ratio
\[
R = \frac{SD}{N_e}.
\]

\subsection{Experimental Design}

Here we vary $R = 0.25, 0.5, 1, 2, 5, 10$.  For this experiment, we fix the number of leaves as $10$.  Therefore $e = 45$.   For each value of $R$, we generate a random species tree via the Yule model first.  Then we generate the set of $1000$ gene trees from the multi-species coalescent model given the species tree.  Therefore, for each $R$, we have a simulated dataset with size $1000$.  

Note that when $R$ is larger we have tighter constraints to gene tree topologies by its species tree. Therefore, we do not have large variance for generating gene trees so that it is easier to estimate missing part of a gene tree.  On the other hand, if we have small $R$, then we have a large variance for gene tree topologies, the coalescent model is getting more like a random process \cite{coalescent}. 

For estimating the performance of our method when we vary the number of leaves missing, we set three different cases: one leaf out of 10 leaves is removed, two leaves out of 10 leaves are removed, and three leaves out of 10 leaves are removed.  For each scenario in terms of $R$ and in terms of the number of leaves removed, we pick random 200 observations from the data set of 1000 trees as a test set.  

To compare the performance of our method, we use the baseline model, i.e., we fill missing values of an ultrametric by taking the mean of  observations with full set of leaves and the multiple linear regression model. For the multiple linear regression model, we set a missing element as a response variable and observed elements in an ultrametric as predictors \cite{8554124}.  

\subsection{Results}

To assess a performance of our method against the baseline and linear regression model, we use the RF distance between an estimated tree $\hat{T}$ and $T$. The results are shown in Table \ref{tab:res1} and Figure \ref{fig:tropVSbase}.  Note that the smaller the RF distance between two trees, the closer their tree topologies are.  When the RF distance is 0, then their tree topologies are the same.  

\begin{table}[h]
\centering
\begin{tblr}{
  cell{1}{2} = {c},
  cell{1}{3} = {c=6}{c},
  cell{2}{2} = {c},
  cell{2}{4} = {c},
  cell{2}{5} = {c},
  cell{2}{6} = {c},
  cell{2}{7} = {c},
  cell{2}{8} = {c},
  cell{3}{1} = {r=3}{},
  cell{6}{1} = {r=3}{},
  cell{9}{1} = {r=3}{},
  vline{1-4} = {1}{},
  vline{3-9} = {2}{},
  vline{3,9} = {3-5,7-8,10-13}{},
  vline{1-3,9} = {6,7, 8, 9}{},
  vline{1-3,9} = {10, 11, 12, 13}{},
  vline{1-3,9} = {1, 2, 3, 4, 5}{},
  hline{1,12} = {-}{0.08em},
  hline{2-3,6,9} = {-}{},
  hline{4-5,7-8,10-11} = {2}{},
}
                 &                   & R  &   &   &   &     &      \\
 \# of leaves removed & Method            & 10 & 5 & 2 & 1 & 0.5 & 0.25 \\
 1 leaf removed & Tropical Metric   &0.76& 2.30 &4.11& 6.11 &7.91 &9.13  \\
                 & Baseline          &2.63 & 6.34&  8.71 &11.36 &12.43& 12.87 \\
                 & Linear Regression & 2.70 & 6.34 & 8.66& 11.26& 12.44 &12.89      \\
2 leaves removed & Tropical Metric   & 0.84& 2.65& 5.36& 6.95& 8.26& 9.20   \\
                 & Baseline          & 2.55 & 6.28 & 8.32 &11.02& 12.55 &13.00 \\
                 & Linear Regression & 2.66 & 6.29&  8.58 &10.96& 12.51 &13.06 \\
3 leaves removed & Tropical Metric   &1.32& 2.95& 6.06& 7.93 &9.85 &9.84 \\
                 & Baseline          & 2.46 & 6.12 & 7.95& 10.80 &12.73& 13.10 \\
                 & Linear Regression &2.56 & 6.16 & 8.05 &10.94 &12.70 &13.23     
\end{tblr}
\caption{These are average RF distances between estimated trees and true trees.  For each category, we infer $200$ trees from $800$ trees.  The smaller the RF distance is, we have better performance.}\label{tab:res1}
\end{table}


\begin{figure}
    \centering
    \includegraphics[width=0.45\textwidth]{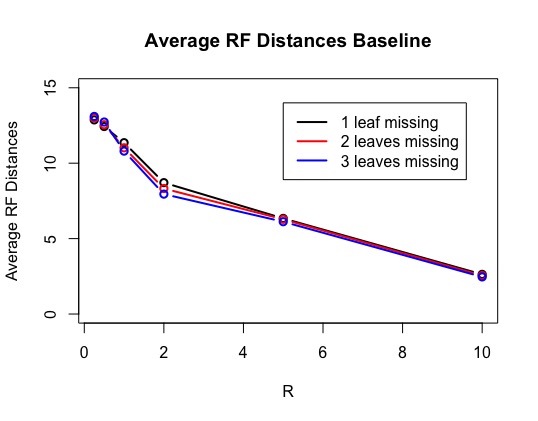}
    \includegraphics[width=0.45\textwidth]{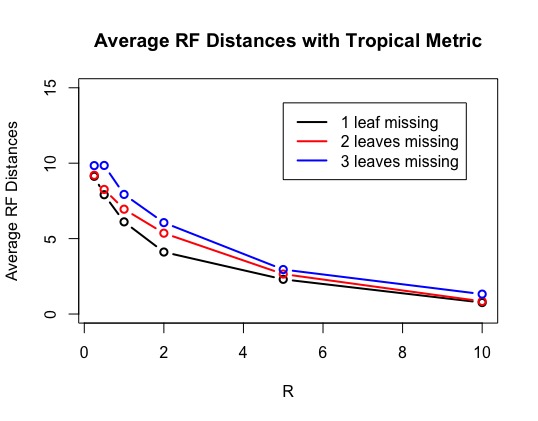}
    \caption{This figure shows performance on the baseline (Left) and our method using a tropical poltyope (Right).  For each category, we infer $200$ trees from $800$ trees.   The x-axis represents the ratio $R$ and the y-axis shows the average RF distances between estimated trees and true trees for $200$ trees. The smaller the RF distance is, we have better performance. }
    \label{fig:tropVSbase}
\end{figure}

\begin{figure}
    \centering
    \includegraphics[width=0.45\textwidth]{baseline.jpeg}
    \includegraphics[width=0.45\textwidth]{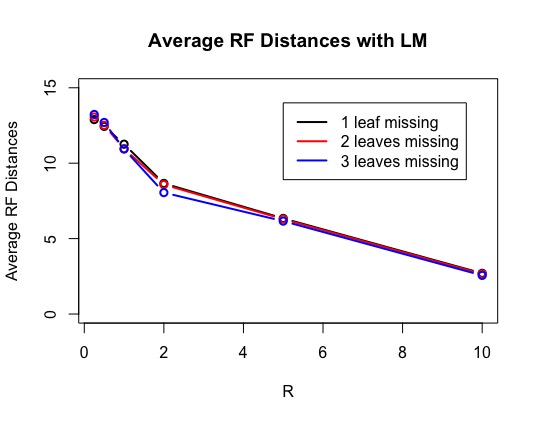}
    \caption{This figure shows performance on the baseline model (Left) and linear regression models (Right). For each category, we infer $200$ trees from $800$ trees.   The x-axis represents the ratio $R$ and the y-axis shows the average RF distances between estimated trees and true trees for $200$ trees. The smaller the RF distance is, we have better performance. As one can see, these results are very close to each other for all $R$.}
    \label{fig:LMVSbase}
\end{figure}

According to our computational experiments with simulated datasets shown in Table \ref{tab:res1} and Figure \ref{fig:tropVSbase}, our method has smaller RF distances in any cases compared to other methods.  It is interesting that the number of leaves removed seems very much affecting the results in general while clearly $R$ affects performances of all three methods we compare.  

If we have only one missing leaf and larger $R$, often the average RF distances between inferred trees and true trees is less than 1 because we have often the condition satisfied in Theorem \ref{th:main} due to very high constraints on tree topologies of gene trees.

\section{Discussion}

In this short paper, we show a novel method to impute a missing part of an equidistant tree on $[m]$ using a tropical polytope, which is an analogue of a linear regression in the setting of tropical geometry.  From simulated data generated from the multi-species coalescent model, we show that this method works very well. In addition we show a condition that the estimate tree and the true tree have at most 4 RF distance (Theorem \ref{th:main}).  

In future, we can investigate applying ``tropical probcipal component analysis (PCA)'' proposed by Yoshida, et al. in \cite{YZZ} to imputation of trees since the classical PCA can be viewed as a multivariate liear regression model with orthogonal projections.


%
\bibliographystyle{plain}
\bibliography{refs}
\end{document}